\documentclass[runningheads,envcountsect,10pt]{llncs}

\usepackage{cite}
\usepackage{amsmath,amssymb,amsfonts}
\usepackage{algorithmic}
\usepackage{graphicx}
\usepackage{textcomp}
\usepackage{xcolor}
\def\BibTeX{{\rm B\kern-.05em{\sc i\kern-.025em b}\kern-.08em
	T\kern-.1667em\lower.7ex\hbox{E}\kern-.125emX}}

\usepackage[ruled,vlined,linesnumbered]{algorithm2e}
\usepackage{diagbox}
\usepackage{subfigure}
\usepackage{multirow}
\usepackage{bm}

\begin{document}
	\title{Sampling-Based Approximate Skyline Calculation on Big Data\thanks{This work was supported by the National Natural Science Foundation of China under grant 61732003, 61832003, 61972110 and U1811461.}}
	
	\titlerunning{Sampling-Based Approximate Skyline Calculation on Big Data}
	\author{Xingxing Xiao\inst{1} \and
		Jianzhong Li\inst{2}
	}
	\authorrunning{X. Xiao, J. Li}
	\institute{
		\email{xiaoxx@hit.edu.cn}\and \email{lijzh@hit.edu.cn\\
			Harbin Institute of Technology, Harbin, Heilongjiang 150001, China
		}
	}
	\maketitle              % typeset the header of the contribution
	\begin{abstract}
		The existing algorithms for processing skyline queries cannot adapt to big data. This paper proposes two approximate skyline algorithms based on sampling. The first algorithm obtains a fixed size sample and computes the approximate skyline on the sample. The error of the first algorithm is relatively small in most cases, and is almost independent of the input relation size. The second algorithm returns an  $(\epsilon,\delta)$-approximation for the exact skyline. The size of sample required by the second algorithm can be regarded as a constant relative to the input relation size, so is the running time. Experiments verify the error analysis of the first algorithm and show that the second algorithm is much faster than the existing skyline algorithms. 
		\keywords{Sampling\and Skyline\and Approximation\and Big Data}
	\end{abstract}
	
	\section{Introduction}
	\label{sec_intro}
	Skyline queries are important in many applications involving multi-criteria decision making. Given a relation $T(A_1, A_2, ..., A_d)$ and a set of skyline criteria $C \subseteq \{A_1, ..., A_d\}$, a skyline query on $T$ is to find a subset of $T$ such that each tuple $t$ in the subset is not dominated by any tuple in $T$, where $t'$ dominates $t$, written as $t'\prec t$, means that $t'.A_i \le t.A_i$ for all $A_i \in C$ and there is an attribute $A_j \in C$ such that $t'.A_j < t.A_j$. $t.A_i$ is the value of tuple $t \in T$ on attribute $A_i$. Skyline queries can also be defined using $\ge$ and $>$. Without loss of generality, this paper only considers the skyline queries defined by $\leq$ and $<$.  The answers to a skyline query are all the potentially best tuples to users, and skyline queries provide good mechanisms for merging user's preferences into queries. 
	
	Studies on skyline queries originated in theoretical computer science area in the last century. Skyline was called as the set of maximals or the pareto set in that time. Many algorithms for finding the maximals were proposed \cite{kung1975finding, bentley1978average, bentley1993fast}. The lowest time complexity of these algorithms is $O(n \log^{d-2}n)$ in the worst case, and $O(n)$ in the average case. However, all the algorithms are based on Divide\&Conquer strategy and assume that their input tuples are stored in main memory. 
	
	Borzsony first introduced skyline queries to the database field\cite{borzsony2001skyline}. It attracted considerable attention to design efficient algorithms for processing skyline queries on relations stored in external storage. Many algorithms have been proposed \cite{borzsony2001skyline, chomicki2003skyline, godfrey2005maximal, bartolini2008efficient}. The lowest time complexity of the algorithms is $O(n^2)$ in the worst case, and $O(n)$ in the average case. 
	
	Nowadays, big data is coming to the force in a lot of applications\cite{gao2020recognizing}. Processing a skyline query on big data in more than linear time is by far too expensive and often even linear time may be too slow. Thus, designing a subliner time algorithm for processing skyline queries becomes a highly concerned research subject. Many index-based algorithms for processing skyline queries have been proposed to achieve the sublinear running time in the average case \cite{borzsony2001skyline, tan2001efficient, kossmann2002shooting, papadias2003optimal, lee2010z, han2012efficient}. However, all the algorithms have serious limitations. Firstly, the algorithms require much time for pre-computation, which is at least $\Omega(n)$. Secondly, they need expensive extra space overhead for indexes. Thirdly, there is much overhead to maintain indexes while the input relations are updated. 
	
	Approximation computation \cite{miao2017complexity, cai2019deletion, miao2020hardness, miao2020computation} of the skyline is the only way to break trough the three limitations. Fortunately, approximate skyline results are enough in many applications. An example of skyline queries is to find restaurants near the workplace that provide delicious foods and excellent services. To get the answer quickly, users can accept approximate skyline results that are the good restaurants but not the best ones. Actually, users prefer to get approximate results in seconds rather than exact results in hours or more in many applications. 
	
	There have been many researches on approximate algorithms for skyline queries \cite{koltun2005approximately, lin2007selecting, tao2009distance, magnani2014taking, soholm2016maximum}, but their goal is to reduce the skyline size and approximate the best subset of $k$ input tuples to represent the skyline under various measures. Moreover, they have higher running time than the precise algorithms for processing skyline queries.
	
	In this paper, we propose two approximate algorithms based on sampling \cite{miao2016complexity}, for processing skyline queries on big data.  The proposed algorithms don't need any extra space or pre-computation overhead. Viewing the skyline as a covering, the error of a approximate algorithm is defined as $|\frac {\mathcal{DN}(Sky) - \mathcal{DN}(\widetilde{Sky})}  {\mathcal{DN}(Sky)} |$, where $\mathcal{DN}(\widetilde{Sky})$ is the number of tuples dominated by the approximate result $\widetilde{Sky}$, and $\mathcal{DN}(Sky)$ is the number of tuples dominated by the exact result $Sky$. If $|\frac {\mathcal{DN}(Sky) - \mathcal{DN}(\widetilde{Sky})} {\mathcal{DN}(Sky)}| \leq \epsilon$, then $\epsilon$ is called as the error bound of the approximate algorithm.
	
	The first algorithm draws a random sample from the input relation at the beginning, and then computes the approximate skyline on the sample. The algorithm has two advantages. First, the expected error of the algorithm is almost independent of the input relation size. Second, the standard deviation of the error is relatively small. These advantages have been verified in experiments.
	
	The second algorithm, DOUBLE, is a random algorithm and returns an  $(\epsilon,\delta)$-approximation for the exact skyline efficiently. The size of sample required by DOUBLE is almost a constant relative to the input relation size. DOUBLE first draws an initial sample, and then computes the approximate skyline on the sample. Afterwards, it judges whether the current result meets the requirement by $Monte$ $Carlo$ $method$. If not, it doubles the sample size and repeats the above process. Otherwise it terminates. Extensive experiments show that DOUBLE involves only constant number of tuples, and is much faster than the existing skyline algorithms.
	
	The main contributions of the paper are listed below.
	
	(1) A baseline approximate algorithm for processing skyline queries is proposed, which is based on a sample of size $m$. The running time of the algorithm is $O(m\log^{d-2}m)$ in the worst case and $O(m)$ in the average case. If $m$ is equal to $n^\frac{1}{k}$ $(k>1)$, the baseline algorithm is in sublinear time. If all skyline criteria are independent of each other, the expected error of the algorithm is 
	\begin{displaymath}
	\overline{\varepsilon} \le \frac{n-m}{n}\sum_{i=0}^{d-1}\frac{(\log (m+1))^i}{i!(m+1)}.
	\end{displaymath}
	And the standard deviation of the error is $o(\overline{\varepsilon})$.
	
	(2) An approximate algorithm, DOUBLE, is proposed to return an  $(\epsilon,\delta)$-approximation for the exact skyline efficiently. The expected sample size required by DOUBLE is $O(\mathcal{M}_{\frac{\epsilon}{3},\delta})$, and the expected time complexity of DOUBLE is  $O(\mathcal{M}_{\frac{\epsilon}{3},\delta}\log^{d-1}\mathcal{M}_{\frac{\epsilon}{3},\delta})$, where $\mathcal{M}_{\frac{\epsilon}{3},\delta}$ is the size of sample required by the baseline algorithm to return an ($\frac{\epsilon}{3},\delta$)-approximation. $\mathcal{M}_{\frac{\epsilon}{3},\delta}$ is almost unaffected by the relation size. 
	
	(3) Extensive experiments are performed on three synthetic data sets and a real data set. The synthetic data sets have reached the terabyte level. The experiments verify the theoretical analysis results of the baseline algorithm, and show that DOUBLE is much faster than the existing skyline algorithms.
	
	The remainder of the paper is organized as follows. Section \ref{sec_problem} provides problem definitions. Section \ref{sec_baseline} describes the baseline algorithm and its analysis. Section \ref{sec_double} presents DOUBLE and its analysis. Section \ref{sec_exp} shows the experimental results. Section \ref{sec_con} concludes the paper. 
	
	\section{Problem Definition}
	\label{sec_problem}
	\subsection{Skyline Definition}
	Let $T(A_1,A_2,...,A_d)$ be a relation with $n$ tuples and $d$ attributes, abbreviated as $T$. In the following, we assume that all attributes are skyline criteria. First, we formally define the dominance relationship between tuples in $T$.
	\begin{definition}
		\label{d_dominance_between_tuples}
		\textbf{(Dominance between Tuples)} Let $t$ and $t'$ be tuples in the relation $T(A_1,A_2,...,A_d)$. $t$ dominates $t'$ with respect to the $d$ attributes of $T$, denoted by $t$\textbf{$\prec$} $t'$, if $t.A_i\le t'.A_i$ for all $A_i \in \{A_1,...,A_d\}$, and $\exists A_j \in \{A_1,...,A_d\}$ such that $t.A_j < t'.A_j$. 
	\end{definition}
	Based on the dominance relationship between tuples, we can define the dominance relationship between sets. In the following, $t\preceq t'$ denotes $t\prec t'$ or $t$ = $t'$ with respect to $d$ attributes of $T$.
	\begin{definition}
		\label{d_dominance_between_sets}
		\textbf{(Dominance between Sets)} A tuple set $Q$ dominates another set $Q'$, denoted by $Q \preceq Q'$, if for each tuple $t'$ in $Q'$, there is a tuple $t$ in $Q$ such that $t\prec t'$ or $t = t'$, i.e. $t\preceq t'$. $Q\preceq \{t\}$ can be abbreviated as $Q\preceq t$.
	\end{definition}
	Now, we define the skyline of a relation.
	
	\begin{definition}
		\label{d_skyline}
		\textbf{(Skyline)} Given a relation $T(A_1,A_2,...,A_d)$, the skyline of $T$ is $Sky(T) = \{t\in T | \forall t' \in $T$, t' \not\prec t \} $. 
	\end{definition}
	
	\begin{definition}
		\label{d_skyline_problem}
		\textbf{(Skyline Problem)} The skyline problem is defined as follows.
		
		Input: a relation $T(A_1,A_2,...,A_d)$.
		
		Output: $Sky(T)$.  
	\end{definition}
	
	The skyline problem can be equivalently defined as following optimization problem.
	
	\begin{definition}
		\textbf{(OP-Sky Problem)} 
		OP-Sky problem is defined as follows.
		
		Input: a relation $T(A_1,A_2,...,A_d)$.
		
		Output: $Q\subseteq T$ such that $|\{ t\in T |Q\preceq t\}|$ is maximized 
		and $\forall t_1, t_2 \in$ $Q, t_1\not\prec t_2$.
	\end{definition}
	
	The following theorem 2.1 shows that the $Skyline~Problem$ is equivalent to the $OP$-$Sky$. 
	
	\begin{theorem}
		\label{t_skylineopt}
		The skyline of $T$ is one of the optimal solutions of the problem $OP_1$. If there is no duplicate tuples in $T$, $Sky(T)$ is the unique optimal solution. 
	\end{theorem}
	
	This paper focus on approximate algorithms for solving the $OP$-$Sky$ problem. The error of an approximate algorithm for an input relation $T$ is defined as $|\frac {\mathcal{DN}(Sky) - \mathcal{DN}(\widetilde{Sky})} {\mathcal{DN}(Sky)}|$, where $\mathcal{DN}(\widetilde{Sky})$ is the number of tuples in $T$ dominated by the approximate solution $\widetilde{Sky}$, and $\mathcal{DN}(Sky)$ is the number of tuples in $T$ dominated by the exact solution $Sky$. If $|\frac {\mathcal{DN}(Sky) - \mathcal{DN}(\widetilde{Sky})} {\mathcal{DN}(Sky)}| \leq \epsilon$, then $\epsilon$ is called as the error bound of the approximate algorithm.  
	
	In the following sections, we will present two approximate algorithms for solving the $OP$-$Sky$ problem.
	
	\section{The Baseline Algorithm and Analysis}
	\label{sec_baseline}
	\subsection{The Algorithm}
	The baseline algorithm first obtains a sample $S$ of size $m$ from the input relation $T$, and then computes the approximate skyline result on $S$. Any existing skyline algorithm can be invoked to compute the skyline of $S$.
	%\vspace{-8mm}
	\begin{algorithm}[!htb]
		\label{a_approximation}
		\caption{The Baseline Algorithm}
		\KwIn{The relation $T(A_1,A_2,...,A_d)$ with $n$ tuples, and the sample size  $m$;}
		\KwOut{$\widetilde{Sky}$, i.e. the approximate skyline of $T$}
		$S$ is the sample of $m$ tuples from $T$\;
		\Return getSkyline($S$). /* getSkyline can be any exact skyline algorithm */\
	\end{algorithm}
	%\vspace{-8mm}
	\subsection{Error Analysis of The Baseline Algorithm}
	To facilitate the error analysis of the baseline algorithm, we assume that the baseline algorithm is based on sampling without replacement. Let $\varepsilon$ be the error of the algorithm, $\overline{\varepsilon}$ be the expected error of the algorithm, and $\sigma^2$ be the variance of the error. 
	
	\subsubsection{The Expected Error} We first analyze the expected error $\overline{\varepsilon}$ of the baseline algorithm. Assume each tuple in $T$ is a $d$-dimensional \textit{i.i.d.} (independent and identically distributed) random variable.
	
	If the $n$ random variables are continuous, we assume that they have the joint probability distribution function $F(v_{1},v_{2},...,v_{d})=F(\overline{V})$, where $ \overline V=(v_1, v_2, ...v_d)$. Let $f(v_{1},v_{2},...,v_{d})=f(\overline{V})$ be the joint probability density function of the random variables. Without loss of generality, the range of variables on each attribute is $[0, 1]$, since the domain of any attribute of $T$ can be transformed to $[0, 1]$. 
	
	\begin{theorem}
		\label{t_continuous_0}
		If all the $n$ tuples in $T$ are $d$-dimensional \textit{i.i.d.} continuous random variables with the distribution function $F(\overline{V})$, then the expected error of the baseline algorithm is  
		\begin{displaymath}
		\overline{\varepsilon}=\frac{n-m}{n}\int_{[0,1]^d}f(\overline{V})(1-F(\overline{V}))^{m}d\overline{V}
		\end{displaymath}
		where $m$ is the sample size, $f(\overline{V})$ is the density function of the variables, and the range of variables on each attribute is $[0, 1]$.
	\end{theorem}
	\begin{proof}
		Due to $\mathcal{DN}(Sky(S))$ = $\mathcal{DN}(S)$ $\le$ $\mathcal{DN}(Sky(T))$ = $n$, where $n$ is the size of the relation $T$, we have 
		\begin{displaymath}
		\varepsilon = |\frac{\mathcal{DN}(Sky(T))-\mathcal{DN}(Sky(S))}{\mathcal{DN}(Sky(T))}| = \frac{n-\mathcal{DN}(S)}{n}
		\end{displaymath}
		
		Let $X_i$ be a random variable for $1 \le i \le n$, and $t_i$ be the $i^{th}$ tuple in $T$. $X_i=0$ if $t_i$ in $T$ is dominated by the sample $S$, otherwise $X_i=1$. Thus, we have $\mathcal{DN}(S)=n-\sum_{i=1}^{n}X_i$ and $\varepsilon=\frac{\sum_{i=1}^{n}X_i}{n}$. By the linearity of expectations, the expected error of the baseline algorithm is $\overline{\varepsilon}  = \frac{\sum_{i=1}^{n}EX_i}{n} = \frac{\sum_{i=1}^{n}Pr(X_i=1)}{n} = Pr(X_i=1)$, where $Pr(X_i=1)$ is the probability that $t_i$ in $T$ is not dominated by $S$. 
		
		Let $Y_i$ be a random variable for $1 \le i \le n$. $Y_i=0$ if $t_i$ in $T$ is picked up into the sample $S$, otherwise $Y_i=1$. According to the conditional probability formula, we have 
		\begin{displaymath}
		Pr(X_i=1) = Pr(Y_i=0)Pr(X_i=1 | Y_i=0) + Pr(Y_i=1)Pr(X_i=1 | Y_i=1)
		\end{displaymath}
		If $t_i$ is selected into in $S$, then it is dominated by $S$. Therefore, we have $Pr(X_i=1 | Y_i=0)$ is equal to $0$. Due to sampling with replacement, $Pr(Y_i=1)$ is equal to $\frac{n-m}{n}$. In short, we have 
		\begin{displaymath}
		Pr(X_i=1) = \frac{n-m}{n} Pr(X_i=1 | Y_i=1)
		\end{displaymath}
		Assume $t_i$ is not selected into $S$. Let $t_i$ have the value $\overline{V} = (v_1,v_2,...,v_d)$. Subsequently, for the $j_{th}$ tuple $t'_j$ in $S$, $t'_j$ satisfies the distribution $F$ and is independent of $t_i$. It is almost impossible that $t_i$ has a value equal to $t'_j$ on an attribute. The probability of $t'_j\prec t_i$ is $F(\overline{V})$. In turn, we have 
		$Pr(t'_j\not\preceq t_i | Y_i=1)=1-F(\overline{V})$.
		
		Because $S$ is a random sample without replacement, all tuples in $S$ are  distinct tuples from $T$. All the tuples in $T$ are independently distributed, so are the tuples in $S$. Therefore, the probability that $S$ doesn't dominate $\{t_i\}$ is
		\begin{displaymath}
		Pr(S\not\preceq \{t_i\} | Y_i=1)=\prod_{j=1}^{m}Pr(t'_j\not\preceq t_i | Y_i=1)=(1-F(\overline{V}))^{m}
		\end{displaymath}
		
		In the analysis above, $\overline V$ is regarded as a constant vector. 
		Since $\overline V$ is a variable vector and has the density function $f(\overline{V})$, we have 
		\begin{displaymath}
		Pr(X_i=1 | Y_i=1)= \int_{[0,1]^d}f(\overline{V})(1-F(\overline{V}))^{m}d\overline{V}
		\end{displaymath} 
		
		Thus the probability that $t_i$ is not dominated by $S$ is 
		\begin{displaymath}
		Pr(X_i=1)=\frac{n-m}{n}\int_{[0,1]^d}f(\overline{V})(1-F(\overline{V}))^{m}d\overline{V}
		\end{displaymath}
		$\hfill\square$ 
	\end{proof}
	
	\begin{corollary}
		\label{c_continuous_0}
		If all the $n$ tuples in $T$ are $d$-dimensional \textit{i.i.d.} continuous random variables, then the expected error of the baseline algorithm is  
		\begin{displaymath}
		\overline{\varepsilon}=\frac{n-m}{n}\frac{\mu_{m+1,d}}{m+1}
		\end{displaymath}
		where $m$ is the sample size and $\mu_{m+1,d}$ is the expected skyline size of a set of $m+1$ $d$-dimensional \textit{i.i.d.} random variables with the same distribution. 
	\end{corollary}
	\begin{proof}
		Let $Q$ be a set of $m+1$ $d$-dimensional \textit{i.i.d.} random variables with the distribution function $F$, then the expected skyline size of $Q$ is
		$\mu_{m+1,d}=(m+1)\int_{[0,1]^d}f(\overline{V})(1-F(\overline{V}))^{m}d\overline{V}$. Based on theorem \ref{t_continuous_0}, we get the corollary. 
		$\hfill\square$
	\end{proof}
	
	If the $n$ random variables are discrete, we assume that they have the joint probability mass function as follows
	\begin{displaymath}
	g(v_{1}, v_{2},..., v_{d}) = Pr(A_1=v_1,A_2=v_2,...,A_d=v_d)
	\end{displaymath}
	Let $G(v_{1},v_{2},...,v_{d})=G(\overline{V})$ be the probability distribution function of the variables. Assume that $\mathcal{V}$ is the set of all tuples in $T$, i.e. all value vectors of the $d$-dimensional variables.
	\begin{theorem}
		\label{t_discrete_0}
		If all the $n$ tuples in $T$ are $d$-dimensional \textit{i.i.d.} discrete random variables with the distribution function $G(\overline{V})$, then the expected error of the baseline algorithm is 
		\begin{displaymath}
		\overline{\varepsilon}=\frac{n-m}{n}\sum_{\overline{V}\in \mathcal{V}}g(\overline{V})(1-G(\overline{V}))^{m}
		\end{displaymath}
		where $m$ is the sample size, $\mathcal{V}$ is the set of all value vectors of the $d$-dimensional variables and $g$ is the mass function.
	\end{theorem}
	The proof is basically the same as theorem \ref{t_continuous_0}, except that duplicate tuples need to be considered. 
	
	Based on theorem \ref{t_continuous_0} and \ref{t_discrete_0}, the relation size has almost no effect on the expected error of the baseline algorithm. Indeed, $m$ is equal to $o(n)$, and $\frac{n-m}{n}$ approaches to $1$ in most cases. 
	
	\begin{corollary}
		\label{c_discrete_0}
		If all the $n$ tuples in $T$ are $d$-dimensional \textit{i.i.d.} discrete random variables, then the expected error of the baseline algorithm is 
		\begin{equation}
		\label{ine1}
		\overline{\varepsilon} \le \frac{n-m}{n}\frac{\mu_{m+1,d}}{m+1}
		\end{equation}
		where $m$ is the sample size and $\mu_{m+1,d}$ is the expected skyline size of a set of $m+1$ $d$-dimensional \textit{i.i.d.} random variables with the same distribution. If there is no duplicate tuples in $T$, then the equality of (\ref{ine1}) holds.
	\end{corollary}
	\begin{proof}
		Let $Q$ be a set of $m+1$ $d$-dimensional \textit{i.i.d.} random variables with the distribution function $G$, then the expected skyline size of $Q$ is
		\begin{align}
		\notag\mu_{m+1,d}=&(m+1)\sum_{\overline{V}\in \mathcal{V}}g(\overline{V})(1-G(\overline{V})+g(\overline{V}))^{m} \\
		\ge& (m+1)\sum_{\overline{V}\in \mathcal{V}}g(\overline{V})(1-G(\overline{V}))^{m} \label{ine2}
		\end{align}
		the equality of (\ref{ine2}) holds if and only if there is no duplicate tuples in $Q$. Based on theorem \ref{t_discrete_0}, the corollary is proved. 
		$\hfill\square$
	\end{proof}
	
	By the analysis of the expected skyline size under stronger assumptions in \cite{godfrey2004skyline}, we further analyze the expected error of the baseline algorithm.
	\begin{definition}
		\label{d_CI}
		\textbf{(Component independence)} $T(A_1,A_2,...,A_d)$ satisfies component independence ($CI$), if all $n$ tuples in $T$  follow the conditions below.
		\begin{enumerate}
			\item \textbf{(Attribute Independence)} the values of tuples in $T$ on a single attribute are statically independent of the values on any other attribute;
			\item \textbf{(Distinct Values)} $T$ is sparse, i.e. any two tuples in $T$ have different values on each attribute.
		\end{enumerate}
	\end{definition}
	
	\begin{theorem}
		\label{t_CIdistribution}
		Under $CI$, the error of the baseline algorithm is unaffected by the specific distribution of $T$.
	\end{theorem}
	\begin{proof}
		If $T$ satisfies component independence, it can be converted into an uniformly and independently distributed set. After conversion, the error of the basline algorithm remains unchanged. The specific conversion process is as follows. Consider each attribute in turn. For the attribute $A_i$, sort tuples in ascending order by values on $A_i$. Then rank $0$ is allocated the lowest value $0$ on $A_i$, and so forth. Rank $j$ is allocated the value $j/n$ on $A_i$.
		$\hfill\square$
	\end{proof}
	
	From \cite{godfrey2004skyline}, we have the following lemma. 
	\begin{lemma}
		Under $CI$, the expected skyline size of $T$ is equal to the $(d-1)^{th}$ order harmonic of $n$, denoted by $H_{d-1,n}$. 
	\end{lemma}
	
	For integers $k > 0$ and integers $n > 0$, $H_{d,n} = \sum_{i=1}^{n} \frac{H_{d-1,i}}{i}$. 
	From \cite{buchta1989average} and \cite{devroye1980note}, we have 
	\begin{displaymath}
	H_{d,n}=\frac{(\log n)^d}{d!}+\gamma\frac{(\log n)^{d-1}}{(d-1)!}+O((\log n)^{d-2}) \le \sum_{i=0}^{d}\frac{(\log n)^i}{i!}
	\end{displaymath}
	, where $\gamma = 0.577...$ is Euler's constant.
	
	From definition \ref{d_CI}, there is no duplicate tuples in $T$ under $CI$. Thus,  based on corollary \ref{c_continuous_0} and \ref{c_discrete_0}, we have the following corollary. 
	\begin{corollary}
		\label{c_CI}
		If the relation $T(A_1,A_2,...,A_d)$ with $n$ tuples satisfies $CI$, then the expected error of the baseline algorithm is   
		\begin{displaymath}
		\overline{\varepsilon}= \frac{n-m}{n(m+1)}H_{d-1,n}\le \frac{n-m}{n(m+1)}\sum_{i=0}^{d-1}\frac{(\log (m+1))^i}{i!}
		\end{displaymath}
		where $m$ is the sample size.
	\end{corollary}
	
	If there are tuples in $T$ with the same values on an attribute and Attribute Independence in definition \ref{d_CI} holds, we have the following corollary. 
	
	\begin{figure}[t]
		\centering
		\includegraphics[width=0.6\linewidth]{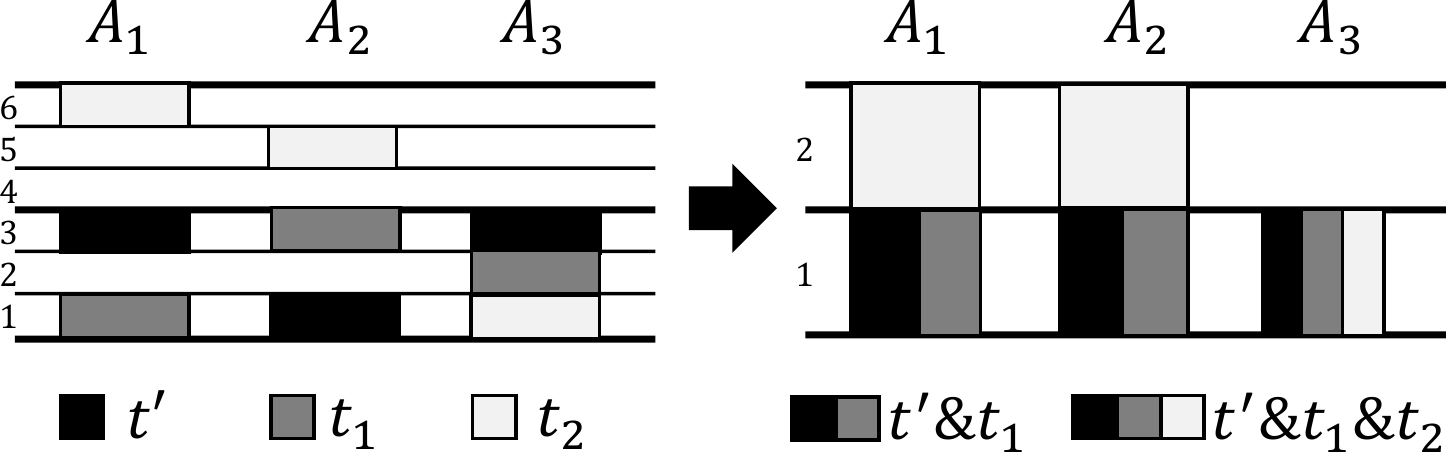}  
		\caption{The Effect of Denseness}
		\label{figex}
	\end{figure}
	
	\begin{corollary}
		\label{c_independence}
		If all attributes in $T(A_1,A_2,...,A_d)$ are independent of each other, then the expected error of the baseline algorithm is 
		\begin{displaymath}
		\overline{\varepsilon} \le \frac{n-m}{n(m+1)}H_{d-1,n} \le \frac{n-m}{n(m+1)}\sum_{i=0}^{d-1}\frac{(\log (m+1))^i}{i!}
		\end{displaymath}
		where $m$ is the sample size.
	\end{corollary}
	
	\begin{proof}
	If distinct values condition in definition \ref{d_CI} doesn't hold, the relation $T$ may be dense, i.e. there are duplicate values of distinct tuples in $T$ on a single attribute. Denseness is equivalent to partitioning values of tuples in a initially sparse relation, into just a few bins (values) over each attribute. There are two cases. First, some tuples initially share no values and have no dominance relationship, but after binned, they do. Second, there may be duplicate tuples that have the same value on each attribute. The expected error of the baseline algorithm is $Pr(S\not\preceq t_i)=Pr(\not\exists t \in S, t\not\prec t'\ or\ t\ne t')$, and both cases decrease it. At this time, the error of the baseline algorithm must be no higher than the value in corollary \ref{c_CI}. 
	$\hfill\square$ \end{proof}

	In figure \ref{figex}, $t'$ in $S$ didn't dominate $t_2$ in $T$ before, but $t'$ dominates $t_2$ after binning. Moreover, $t'$ in $S$ and $t_1$ in $T$ were comparable before, but  $t'$ is equal to $t_1$ after binning. 
	
	\begin{corollary}
		If all attributes in $T(A_1,A_2,...,A_d)$ are independent of each other, with sample size $m$ equal to $n^{\frac{1}{k}}-1 (k> 1)$, then the expected error of the baseline algorithm is 
		\begin{displaymath}
		\overline{\varepsilon} \le \frac{n-n^{\frac{1}{k}}+1}{n}\sum_{i=0}^{d-1}\frac{(\log n)^i}{k^in^{\frac{1}{k}}i!}
		\end{displaymath}
		where $m$ is the sample size.
	\end{corollary}
	
	\subsubsection{Variance of The Error.}
	We assume that each tuple in $T$ is a $d$-dimensional \textit{i.i.d} random variable and $T$ satisfies component independence ($CI$). Without losing generality, all random variables are uniformly distributed over $[0,1]^d$. 
	
	\begin{theorem}
		\label{t_variance} 
		If the relation $T(A_1,A_2,...,A_d)$ with $n$ tuples satisfies $CI$, then $\sigma^2 = O(\frac{\mu_{m,d}}{m^2})$, and $\sigma = O(\frac{\sqrt{\mu_{m,d}}}{m}) = o(\overline{\varepsilon})$.
	\end{theorem}
	\begin{proof}
		Let $X_i$ be a random variable for $1 \le i \le n$. $X_i=0$ if  $t_i$ in $T$ is dominated by the sample $S$, otherwise $X_i=1$. From the proof in theorem \ref{t_continuous_0}, we have 
		
		\begin{align*}
		\sigma^2 &= D(\varepsilon) = D(\frac{\sum_{i=1}^{n}X_i}{n})=D(\sum_{i=1}^{n}X_i)/n^2\\
		&=(E(\sum_{i=1}^{n}X^2_i)+E(\sum_{i\ne j}X_i X_j)-E^2(\sum_{i=1}^{n}X_i))/n^2\\
		&=\frac{1}{n}Pr(X_i=1)+\frac{n-1}{n}Pr_{i\ne j}(X_i=X_j=1) -{Pr}^2(X_i=1).
		\end{align*}
		
		Assume the $t_i$ in $T$ has the value $\overline{U}=(u_1,u_2,...,u_d)$ and the $j_{th}$ tuple $t_j$ has the value $\overline{V}=(v_1,v_2,...,v_d)$. Let
		$(\eta)$ be 
		\begin{displaymath}
		\{(\overline{U},\overline{V}) | \overline{U} \in [0,1]^d,\overline{V}\in [0,1]^d, u_1\le v_1, ..., u_{\eta}\le v_{\eta}, v_{\eta+1}< u_{\eta+1}, ..., v_d<u_d \}.
		\end{displaymath}
		%$\{(\overline{U},\overline{V})$ $|$ $0\le u_1\le v_1\le 1,$ $...$, $0\le u_{\eta}\le v_{\eta}\le 1,$ $0\le v_{\eta+1}< u_{\eta+1}\le 1,$ $...,$ $0\le v_d<u_d\le 1 \}$. 
		$(\eta)$ represents the set of all possible $(\overline{U},\overline{V})$, in which $\overline{U}$ has values no more than $\overline{V}$ on the first $\eta$ attributes and has higher values on the subsequent attributes. Then we have 
		\begin{align*}
		&\quad Pr_{i\ne j}(X_i=X_j=1)\\
		&=\frac{(n-m)(n-1-m)}{n(n-1)}\sum_{\eta=0}^{d}\tbinom{d}{\eta} \int_{(\eta)} (1-\prod_{i=1}^{d}u_i-\prod_{i=1}^{d}v_i +\prod_{i=1}^{\eta}u_i\prod_{i=\eta+1}^{d}v_i)^m d\overline{U}d\overline{V}.
		\end{align*}
		In the above equation, $\frac{(n-m)(n-1-m)}{n(n-1)}$ is the probability that two distinct tuples both are not selected into the sample. Based on \cite{bai2008variance}, we have 
		\begin{displaymath}
		\sum_{\eta=1}^{d-1}\tbinom{d}{\eta} \int_{(\eta)} (1-\prod_{i=1}^{d}u_i-\prod_{i=1}^{d}v_i+\prod_{i=1}^{\eta}u_i\prod_{i=\eta+1}^{d}v_i)^m d\overline{U}d\overline{V}=\frac{\mu^2_{m+2,d}+O(\mu_{m+2,d})}{(m+1)(m+2)}.
		\end{displaymath}
		Thus,
		\begin{align}
		&\quad Pr_{i\ne j}(X_i=X_j=1)\notag\\
		&=\frac{(n-m)(n-1-m)}{n(n-1)} (2\int_{[0,1]^d}  (1-\prod_{i=1}^{d}v_i)^m \prod_{i=1}^{d}v_i d\overline{V} +  \frac{\mu^2_{m+2,d}+O(\mu_{m+2,d})}{(m+1)(m+2)}) \label{e_2}\\
		&=\frac{(n-m)(n-1-m)}{n(n-1)} \frac{2\mu_{m+2,d}(2)+\mu^2_{m+2,d}+O(\mu_{m+2,d})}{(m+1)(m+2)}. \label{e_3}
		\end{align}
		Equation (\ref{e_2}) is based on variable substitution. In (\ref{e_3}), $\mu_{n,d}(r)$ denotes the expected size of the $r_{th}$ layer skyline of $T$, where the $r^{th}$ layer skyline of $T$ is the set of tuples in $T$ that are dominated by exactly $r-1$ tuples in $T$, and its expected size is equal to 
		\begin{equation}
		\mu_{n,d}(r)=n\tbinom{n-1}{r-1} \int_{[0,1]^d}  (1-\prod_{i=1}^{d}v_i)^{n-r} (\prod_{i=1}^{d}v_i)^{r-1} d\overline{V}.\notag
		\end{equation}
		Due to $Pr(X_i=1)=\frac{n-m}{n} \frac{\mu_{m+1,d}}{m+1}$, we have 
		\begin{align}
		&\ \ \ \ D(\sum_{i=1}^{n}X_i)\notag\\
		&=(n-m)\frac{\mu_{m+1,d}}{m+1}-\frac{(n-m)(n+1)}{(m+1)^2(m+2)}\mu^2_{m+1,d} +\frac{(n-m)(n-1-m)}{(m+1)(m+2)}(\mu_{m+2,d}+ \notag\\
		&\quad \mu_{m+1,d})(\mu_{m+2,d}-\mu_{m+1,d}) + \frac {(n-m)(n-1-m)} {(m+1)(m+2)} (2\mu_{m+2,d}(2) + O(\mu_{m+2,d}))\notag\\
		&=(n-m)\frac{\mu_{m+1,d}}{m+1}-\frac{(n-m)(n+1)}{(m+1)^2(m+2)}\mu^2_{m+1,d} + \frac{(n-m)(n-1-m)}{(m+1)(m+2)}O(\mu_{m+2,d}). \label{e_4}
		\end{align}
		Equation (\ref{e_4}) holds because $\mu_{m+2,d}-\mu_{m+1,d}\le 1$ and $\mu_{m+2,d}(2)\le\mu_{m+2,d}$. With $\mu_{m+1,d} \le m+1$, it is true that $D(\sum_{i=1}^{n}X_i) = O(\frac{n^2}{m^2}\mu_{m,d})$. $\hfill\square$
	\end{proof}
	
	\subsection{Analysis of The Time Complexity}
	\begin{theorem}
		If getSkyline in step 2 is based on FLET \cite{bentley1993fast}, then the time complexity of the baseline algorithm is $O(m\log^{d-2}m)$ in the worst case, and $O(m)$ in the average case. 
	\end{theorem}
	\begin{proof}
		Since the time complexity of FLET \cite{bentley1993fast} is $O(n{\log}^{d-2}n)$ in the worst case, and $O(n)$ in the average case, step 2 of the algorithm needs $O(m\log^{d-2}m)$ time. Thus, the time complexity of the algorithm is $O(m\log^{d-2}m)$ because that step 1 of the algorithm needs $O(m)$ time.  $\hfill\square$
	\end{proof}
	\begin{corollary}
		\label{c_result}
		If sample size $m$ equal to $n^{\frac{1}{k}}(k>1)$, then the running time of the baseline algorithm is $O(n^{\frac{1}{k}} \log^{d-2}n^{\frac{1}{k}})$ in the worst case, and $O(n^\frac{1}{k})$ in the average case. 
	\end{corollary}
	Corollary \ref{c_result} tells that the baseline algorithm is in sublinear time if the sample size $m$ is equal to $n^{\frac{1}{k}}(k>1)$.

	\section{DOUBLE and Analysis}
	\label{sec_double}
	In this section, we devise a sampling-based algorithm, DOUBLE, to return an $(\epsilon,\delta)$-approximation efficiently for the exact skyline of the given relation $T$. It first draws an initial sample of size $s_I$ (line 1). The value of $s_I$ can be set to any positive integer. Afterwards, DOUBLE computes the approximate skyline result on the sample (line 2), and then verifies the error $\varepsilon$ of the current result $\widetilde{Sky}$ (lines 3-4). If it is guaranteed that $Pr(\varepsilon\le \epsilon)$ is at least $1-\delta$, then DOUBLE terminates (line 8). Otherwise, it doubles the sample size and repeats the above process (lines 5-7). 

	\begin{algorithm}[htb]
		\SetKwBlock{DoWhile}{Do}{end}
		\caption{DOUBLE}
		\KwIn{\\
			\qquad $T$: the input relation with $n$ tuples and $d$ attributes;\\
			\qquad $\epsilon$: the error bound;\\
			\qquad $\delta$: the error probability;}
		\KwOut{\\
			\qquad an ($\epsilon,\delta$)-approximation $\widetilde{Sky}$ of the skyline of $T$;}
		$m=s_I$, and $S[1,...,m]$ is the sample of $m$ tuples\;
		$\widetilde{Sky}$ = getSkyline($S[1,...,m]$)\;
		$\hat{\varepsilon}$=verifyError($\widetilde{Sky}$)\;
		\textbf{While} ({$\hat{\varepsilon}>\frac{2\epsilon}{3}$})
		\DoWhile{
			$m=2m$, and $S[\frac{m}{2}+1,...,m]$ is the sample of $m/2$ tuples\;
			$\widetilde{Sky}$ = mergeSkyline($\widetilde{Sky}$, getSkyline($S[\frac{m}{2}+1,...,m]$))\;
			$\hat{\varepsilon}$=verifyError($\widetilde{Sky}$)\;
		}
		\Return $\widetilde{Sky}$\;
	\end{algorithm}

	\begin{algorithm}[htb]
		\SetKwProg{verifyError}{verifyError}{}{end}
		\NoCaptionOfAlgo
		\caption{\textbf{Subroutines 1}}
		\setcounter{AlgoLine}{0}
		\verifyError{($\widetilde{Sky}$)}{
			$count=0$, and $s_v$ = $\lceil \frac{18(\ln\log_2 n+\ln(\frac{1}{\delta}))}{\epsilon} \rceil$\;
			$V$ is the sample of $s_v$ tuples\;
			\For{each tuple $t$ in $V$}
			{
				\If{$t$ is not dominated by $\widetilde{Sky}$}{
					count+=1\;
				}
			}
			\Return $count/s_v$\;
		}
	\end{algorithm}
	
	DOUBLE judges whether $\varepsilon$ meets the requirement by $Monte$ $Carlo$ $method$. In the subroutine verifyError, $s_v$ is the sample size for each verification, and is equal to $\lceil \frac{18(\ln\log_2 n+\ln(\frac{1}{\delta}))}{\epsilon} \rceil$ (line 2). DOUBLE first obtains a random sample $V$ of size $s_v$ (line 3). Then it counts and returns the proportion of tuples in $V$ not dominated by the approximate result $\widetilde{Sky}$ (lines 4-7), which is denoted by $\hat{\varepsilon}$. If $\hat{\varepsilon}\le \frac{2\epsilon}{3}$, it is guaranteed that the error of $\widetilde{Sky}$ is not higher than the error bound $\epsilon$ with a probability no less than $1-\delta$. In the following, we prove the above in detail.
	
	\subsection{Error Analysis of DOUBLE}
	
	Let $q$ be the total number of times to invoke verifyError. For $1\le j\le q$, $m_j$ (respectively, $\varepsilon_j$) denotes the value of $m$ (respectively, $\varepsilon$) when verfyError is being invoked for the $j^{th}$ time. $\widetilde{Sky}_j$ is defined in a similar way. $\hat{\varepsilon}_j$ is the value returned by the $j^{th}$ invocation of verifyError. Then we have the following theorem.
	
	\begin{theorem}
		\label{t_incre1}
		For the $j^{th}$ invocation of verifyError, if $\varepsilon_j>\epsilon$, then $Pr(\hat{\varepsilon}_j\le \frac{2}{3}\epsilon)<\frac{\delta}{\log_2 n}$.
	\end{theorem}
	\begin{proof}
		Let $X_i$ be a random variable for $1\le i\le s_v$. For the $j^{th}$ invocation of verifyError, $X_i=0$ if tuple $t_i$ in $V$ is dominated by the approximate result $\widetilde{Sky}_j$, otherwise $X_i=1$. Obviously, $\hat{\varepsilon}_j$ is equal to $\frac{1}{s_v}\sum_{i=1}^{s_v}X_i$. According to the definition of $\varepsilon$, $E(\hat{\varepsilon}_j) = Pr(X_i=1)=\varepsilon_j$. 
		
		By the Chernoff bound, we have 
		\begin{displaymath}
		Pr(\hat{\varepsilon}_j \le \frac{2}{3}\varepsilon_j)\le \frac{1}{e^{s_v\varepsilon_j/18}}.
		\end{displaymath}
		With $\varepsilon_j> \epsilon$ and $s_v\ge \frac{18(\ln\log_2 n+\ln(\frac{1}{\delta}))}{\epsilon}$, we get the theorem.
		$\hfill\square$
	\end{proof}
	
	$\widetilde{Sky}_q$ is the final result returned by DOUBLE. Next, we show that $\widetilde{Sky}_q$ is an ($\epsilon,\delta$)-approximation of the exact skyline.
	
	\begin{corollary}
		If DOUBLE terminates normally, it returns an ($\epsilon,\delta$)-approximation $\widetilde{Sky}_q$, i.e. the error $\varepsilon_q$ of $\widetilde{Sky}_q$ satisfies 
		\begin{displaymath}
		Pr(\varepsilon_q \le \epsilon) =  Pr(|\frac{\mathcal{DN}(\widetilde{Sky}_q)-\mathcal{DN}(Sky)}{\mathcal{DN}(Sky)}|\le \epsilon) \ge 1-\delta.
		\end{displaymath}
	\end{corollary}
	\begin{proof}
		DOUBLE finally returns an ($\epsilon,\delta$)-approximation $\widetilde{Sky}_q$, if and only if, for any positive integer $j<q$, the $j^{th}$ invocation of verifyError with the error $\varepsilon_j>\epsilon$ must return an estimated value $\hat{\varepsilon}_j>\frac{2\epsilon}{3}$. The number of times to invoke verifyError is at most $\log_2n$. Based on theorem \ref{t_incre1}, the probability in this corollary is at least $(1-\frac{\delta}{\log_2n})^{\log_2n}> 1-\delta$. 
		$\hfill\square$
	\end{proof}
	
	\subsection{Analysis of Sample Size and Time Complexity}
	$m_q$ is the final value of $m$. Assume $\mathcal{M}_{\epsilon,\delta}$ is the size of sample required by the baseline algorithm running on $T$ to return an ($\epsilon,\delta$)-approximation. Based on analysis in section \ref{sec_baseline}, $\mathcal{M}_{\epsilon,\delta}$ and $m_q$ are almost unaffected by the relation size $n$. Here we analyze the relationship between $\mathcal{M}_{\epsilon,\delta}$ and $m_q$.
	
	\begin{theorem}
		\label{t_incre2}
		If $\delta\le 1/8$, the expected value of $m_q$ is $O(\mathcal{M}_{\frac{\epsilon}{3},\delta})$. 
	\end{theorem}
	\begin{proof}
		If $m_q$ is less than $\mathcal{M}_{\frac{\epsilon}{3},\delta}$, then the theorem holds. Otherwise, for the $j^{th}$ invocation of verifyError with $\mathcal{M}_{\frac{\epsilon}{3},\delta} \le m_j< 2\mathcal{M}_{\frac{\epsilon}{3},\delta}$, the error $\varepsilon_j$ of $\widetilde{Sky}_j$ satisfies $Pr(\varepsilon_j \le \frac{\epsilon}{3})\ge 1-\delta$. Under the condition $\varepsilon_j\le \frac{\epsilon}{3}$, the probability of $\hat{\varepsilon}_j\le \frac{2\epsilon}{3}$ is at least $1-(\frac{\delta}{\log_2 n})^2$. Thus the probability of $q\le j$ is at least $(1-\delta)(1-(\frac{\delta}{\log_2 n})^2)>1-2\delta$. Similarly, for any $k>j$, the probability of $q>k$ is less than  $2\delta$. Thus the expected value of $m_q$ is at most 
		\begin{displaymath}
		2\mathcal{M}_{\frac{\epsilon}{3},\delta} + 2^2\mathcal{M}_{\frac{\epsilon}{3},\delta}\times 2\delta +... +2^i\mathcal{M}_{\frac{\epsilon}{3},\delta}\times (2\delta)^{i-1}+...
		\end{displaymath}
		With $\delta\le 1/8$, it is $O(\mathcal{M}_{\frac{\epsilon}{3},\delta})$. 
		$\hfill\square$
	\end{proof}
	
	Based on analysis in section \ref{sec_baseline}, $\mathcal{M}_{\frac{\epsilon}{3},\delta}$ is up-bounded by  $O(\mathcal{M}_{\epsilon,\delta})$ in most cases. Thus, the sample used by DOUBLE has the same order of magnitude as the baseline algorithm. Hereafter, we analyze the time complexity of DOUBLE on $m_q$.
	
	\begin{theorem}
		\label{t_incre3}
		If getSkyine and mergeSkyline are based on FLET \cite{bentley1993fast}, then the time complexity of DOUBLE is $O(m_q\log^{d-1} m_q+(\ln\log n+\ln\frac{1}{\delta})m_q\log m_q)$, where $m_q$ is the final value of $m$ in DOUBLE.
	\end{theorem}
	\begin{proof}
		Except for verifyError, the algorithm process is completely equivalent to SD\&C \cite{bentley1978average}. Therefore, the total running time of getSkyline and mergeSkyline is $T'(m_q,d) = 2T'(m_q/2,d) + M(m_q,d)$, where $M(m_q,d)$ is $O(m_q\log^{d-2} m_q)$. Thus we have $T'(m_q,d)=O(m_q\log^{d-1} m_q)$. The number of verifications is $O(\log m_q)$. Based on corollary \ref{c_continuous_0} and \ref{c_discrete_0}, the size of $\widetilde{Sky}$ is $O(\epsilon m_q)$. Therefore, due to $\epsilon s_v= O(\ln\log n+\ln\frac{1}{\delta})$, the total running time of verifyError is $O(\epsilon m_qs_v\log m_q)$ =  $O((\ln\log n+\ln\frac{1}{\delta})m_q\log m_q))$. Finally, the time complexity of the algorithm is $O(m_q\log^{d-1}m_q + (\ln\log n+\ln\frac{1}{\delta})m_q\log m_q)$. 
		$\hfill\square$
	\end{proof}
	
	Even if $n$ is up to $2^{70}$, $\ln\log_2 n$ is less than $5$. Without loss of generality, the time complexity of DOUBLE is $O(m_q\log^{d-1} m_q)$. Through a proof similar to theorem \ref{t_incre2}, we get the following corollary. 
	
	\begin{corollary}
		For $\mathcal{M}_{\frac{\epsilon}{3},\delta}\ge \frac{1}{2^{\frac{1}{d-1}}-1}$ and $\delta\le 1/16$, if getSkyine and mergeSkyline are based on FLET \cite{bentley1993fast}, then the expected time complexity of DOUBLE is $O(\mathcal{M}_{\frac{\epsilon}{3},\delta}\log^{d-1}\mathcal{M}_{\frac{\epsilon}{3},\delta})$.
	\end{corollary}

\section{EXPERIMENTAL RESULTS}
\label{sec_exp}
\subsection{Experimental Settings}
We implemented the two approximation algorithms in C++, and then ran the algorithms on Dell OptiPlex-7500(4 Cores, 8 Threads 3.6GHz i7 CPU + 16G memory + 64 bit Linux). The experiments consider the external storage and all data is stored in Seagate STDR4000(4TB). The experimental results are computed by averaging 20 executions of the approximation algorithm. 

Similar to the experimental design in \cite{han2012efficient}, there were four data sets used in experiments, three synthetic data sets (independent distribution, correlated distribution and anti-correlated distribution) and a real data set. In synthetic data sets, the tuple size is 128 bytes, and there are 8 numeric attributes and one redundancy attribute. For independent distribution, the values of tuples on each attribute are uniformly and independently distributed. For correlated distribution, the values on the first two attributes are generated with Pearson Correlation Coefficient (PCC) 0.5, and the others are uniform and independent.  For anti-correlated distribution, the values on the first two attributes are generated with PCC -0.5. The real data set comes from UCI Machine Learning Repository \cite{frank2010uci}, and are kinematic properties measured by the particle detectors in the accelerator, in which the tuple size is 1024 bytes. Each tuple has 29 numeric attributes, and the remaining space is redundant characters. Experiments set the disk-page size to $8192$ bytes. 

To verify the performance of DOUBLE, we compared it with LESS \cite{godfrey2005maximal} and BNL \cite{borzsony2001skyline}. We do not consider the index-based skyline algorithm. The practicability of the index-based algorithm is severely limited due to its high pre-calculation cost and space overhead. Consider the generic skyline algorithms, roughly can be divided into scan-based algorithms (such as BNL \cite{borzsony2001skyline}, SFS\cite{chomicki2003skyline}, SalSa\cite{bartolini2008efficient} and LESS \cite{godfrey2005maximal}) and divide\&conquer algorithms (such as D\&C \cite{kung1975finding}, LD\&C \cite{bentley1978average}, FLET \cite{bentley1993fast}, and SD\&C \cite{borzsony2001skyline}). Most of the existing divide\&conquer algorithms are not external, and their actual performance is much disappointing. Even SDC, a external divide\&conquer algorithm, is also inferior to LESS and BNL in actual performance. For scan-based algorithms, LESS combines the advantages of SFS and BNL, and has lower I/O cost than SalSa. Without loss of generality, $s_I$ is equal to $s_v$.

\subsection{Experiment 1: The Analysis about The Baseline Algorithm}
\begin{table}[t]
	\caption{Error of The Baseline Algorithm}
	\scriptsize
	%\centering
	\label{tab_ex1}
	\subtable[Independent Distribution]{
		\begin{tabular}{|c|l|l|l|l|l|} \hline
			\diagbox{d}{$\varepsilon$}{m} & 100 & 1000 & 10000 & 100000 & 1000000\\ \hline
			2 & 0.0618 & 0.00752 & 9.61e-04 & 1.32e-04 & 1.41e-05\\ \hline
			3 & 0.144 & 0.0303 & 0.00508 & 7.80e-04 & 1.07e-04\\ \hline
			4 & 0.266 & 0.0718 & 0.0165 & 0.00295 & 4.98e-04\\ \hline
			5 & 0.430 & 0.162 & 0.0419 & 0.00934 & 0.00187\\ \hline
		\end{tabular}
		\label{tab_independent}
	}
	\subtable[Correlated Distribution]{
		\begin{tabular}{|c|l|l|l|l|l|} \hline
			\diagbox{d}{$\varepsilon$}{m} & 100 & 1000 & 10000 & 100000 & 1000000\\ \hline
			2 & 0.0328 & 0.00388 & 5.08e-04 & 7.59e-05 & 6.29e-06\\ \hline
			3 & 0.105 & 0.0202 & 0.00277 & 4.62e-04 & 5.63e-05\\ \hline
			4 & 0.205 & 0.0525 & 0.0113 & 0.00194 & 3.23e-04\\ \hline
			5 & 0.371 & 0.123 & 0.0307 & 0.00665 & 0.00126\\ \hline
		\end{tabular}
		\label{tab_correlated}
	}
	\subtable[Anti-Correlated Distribution]{
		\begin{tabular}{|c|l|l|l|l|l|} \hline
			\diagbox{d}{$\varepsilon$}{m} & 100 & 1000 & 10000 & 100000 & 1000000\\ \hline
			2 & 0.1131 & 0.03212 & 9.52e-03 & 2.96e-03 & 9.39e-04\\ \hline
			3 & 0.227 & 0.0681 & 0.0201 & 0.00617 & 0.00189\\ \hline
			4 & 0.364 & 0.131 & 0.0420 & 0.0128 & 0.00385\\ \hline
			5 & 0.515 & 0.222 & 0.0792 & 0.0256 & 0.00792\\ \hline
		\end{tabular}
		\label{tab_anticorrelated}
	}
	\subtable[Real Data]{
		\begin{tabular}{|c|l|l|l|l|l|} \hline
			\diagbox{d}{$\varepsilon$}{m} & 100 & 1000 & 10000 & 100000 & 1000000\\ \hline
			2 & 0.0539 & 0.00952 & 1.16e-03 & 1.54e-04 & 1.75e-05 \\ \hline
			3 & 0.153 & 0.0327 & 0.00554 & 8.31e-04 & 1.07e-04 \\ \hline
			4 & 0.317 & 0.0884 & 0.0183 & 0.00353 & 5.52e-04 \\ \hline
			5 & 0.437 & 0.166 & 0.0419 & 0.00862 & 0.00155 \\ \hline
		\end{tabular}
		\label{tab_realdata}
	}
\end{table}

\begin{figure}[t]
	\centering
	\subfigure[Predicted Error] {
		\label{fig1a}
		\includegraphics[width=0.45\linewidth]{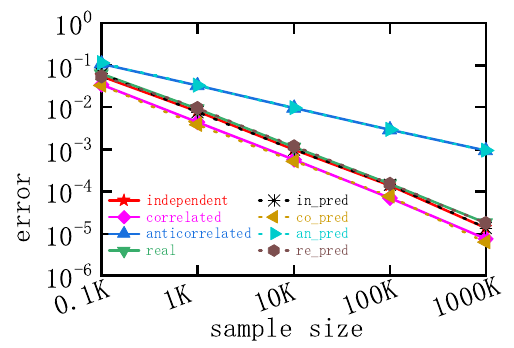}  
	} 
	\subfigure[Variance] {
		\label{fig1b}
		\includegraphics[width=0.45\linewidth]{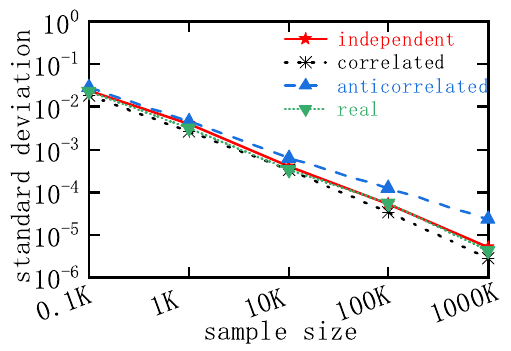}  
	} 
	\caption{Error of Approximation Algorithm}
	\label{fig1}
\end{figure}

In experiment 1, we verified the analysis about the baseline algorithm in section \ref{sec_baseline}. All four data sets were used. The skyline criterion sizes considered are 2, 3, 4, 5. The considered sample sizes are 0.1K, 1K, 10K, 100K, and 1000K (K=$10^3$). The data volumes of synthetic data sets considered are 1TB, whereas that of the real data set is 10GB, i.e. the maximum size of the real data set. It is found that the error of the baseline algorithm is not affected by data volume in experiments. Table \ref{tab_ex1} shows the error of the baseline algorithm on the four data sets, with varying sample sizes and skyline criterion sizes. The values in the table are the average obtained from multiple trials. As shown in table \ref{tab_independent}, in the case of two skyline criteria and independent distribution, the error of the baseline algorithm with size $1000$ is less than $0.01$. Even with $d=5$, the error is still $0.162$. In order to make the error lower than $0.01$, the algorithm only needs a sample of size less than 100000, which is much smaller than $n$ on big data. Table \ref{tab_correlated} shows that under correlated distribution, the error is relatively lower. Table \ref{tab_anticorrelated} demonstrates that even with anti-correlated distribution, a certain size of sample can achieve a small enough error. Under the case $d=5$, the error of the baseline algorithm with the sample size $10000$ is $0.0792$. On the real data set, the error with the sample size $10000$ is $0.0419$ even with $d=5$, shown in table \ref{tab_realdata}. On any data set, the algorithm with a sample of size 1000000 has a error less than $0.01$. One million is relatively small on big data. In addition, we can look up the table \ref{tab_ex1} to obtain the required sample size for the target error. 

Figure \ref{fig1} shows a further analysis of the error, verifying the analysis in section \ref{sec_baseline}. We can use the skyline proportion of the sample to approximate $\frac{\mu_{m+1,d}}{m+1}$, and then use it to predict the error of the algorithm, based on corollary \ref{c_continuous_0} and \ref{c_discrete_0}. Figure \ref{fig1a} compares the predicted and real errors. For each data set, the curves of the predicted and real errors closely fit. The fact that the predicted and real errors are nearly equal shows the correctness of corollary \ref{c_continuous_0} and \ref{c_discrete_0}. As shown in figure \ref{fig1b}, for all data sets, as the sample size increases, the standard deviation of the error decreases significantly. If the sample size is equal to 10K, the standard deviation is less than $0.001$ for all data sets. With a moderate size sample, the standard deviation is relatively small.

\subsection{Experiment 2: The Analysis about DOUBLE}
\subsubsection{Experiment 2.1: The Effect of Data Volume}
\begin{figure}[t]
	\centering
	\subfigure[Execution Time] {
		\label{fig2a}
		\includegraphics[width=0.45\linewidth]{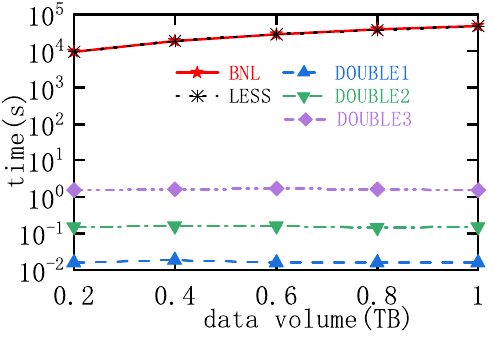}  
	}     
	\subfigure[I/O Overhead] {
		\label{fig2b}
		\includegraphics[width=0.45\linewidth]{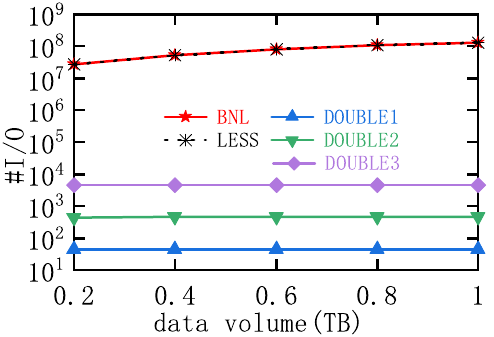}  
	}     
	\subfigure[Used Sample Size] { 
		\label{fig2c}     
		\includegraphics[width=0.45\linewidth]{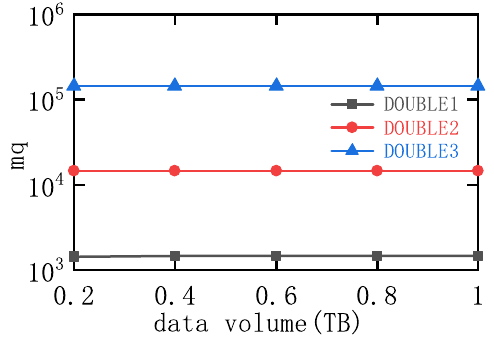}     
	}    
	\subfigure[Result's Representative Ratio] { 
		\label{fig2d}     
		\includegraphics[width=0.45\linewidth]{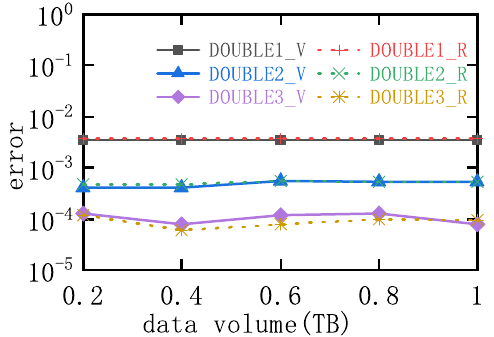}     
	}
	\caption{The Effect of Data Volume}
\end{figure}
In experiment 2.1, we validate the effect of the data volume on DOUBLE with $2$ skyline criteria. All algorithms are executed on independently distributed data set. The data volumes considered are 0.2TB, 0.4TB, 0.6TB, 0.8TB and 1.0TB. The considered error bounds for DOUBLE are 0.1, 0.01, and 0.001. DOUBLE with error bound 0.1 (respectively, 0.01 and 0.001) is represented by DOUBLE1 (respectively, DOUBLE2 and DOUBLE3). As shown in figure \ref{fig2a}, the execution time of DOUBLE is basically unaffected by the data volume, while these of LESS and BNL increase linearly with it. Indeed, the running time of DOUBLE is $O(1)$, relative to the data volume. Moreover, DOUBLE1 is nearly 6 orders of magnitude faster than LESS and BNL. Even DOUBLE2 (respectively, DOUBLE3) is 5 (respectively, 4) orders of magnitude faster on average. As shown in figure \ref{fig2b}, I/O overheads of BNL and LESS increase linearly with the data volume. They has to read each tuple in the table at least once. However, for DOUBLE, I/O cost has nothing to do with the data volume. The size of sample required by DOUBLE is almost independent of the number of tuples in data set. The I/O overhead of DOUBLE1 is 6 orders of magnitude less than these of LESS and BNL. The I/O overhead ratios between different algorithms is basically the same as the runtime ratios. It can be inferred that I/O overhead is dominant relative to CPU overhead. Figure \ref{fig2c} shows the actual sample size required by DOUBLE with respect to the data volume, i.e. $m_q$, which is found to be fixed. At this point, the initial sample is sufficient to cope with the target error bound. However, as the given error bound increases, the required sample size increases significantly. Figure \ref{fig2d} compares the verified error with the real error of the returned set. The curves of the two are relatively close. Both of them are less than the given error bound, and are hardly affected by the data volume. Moreover, the real error are excellent and the approximate skyline is sufficiently used to approach the real skyline. 

\subsubsection{Experiment 2.2: The Effect of Skyline Criterion Size}
\begin{figure}[t]
	\centering
	\subfigure[Execution Time] {
		\label{fig3a}
		\includegraphics[width=0.45\linewidth]{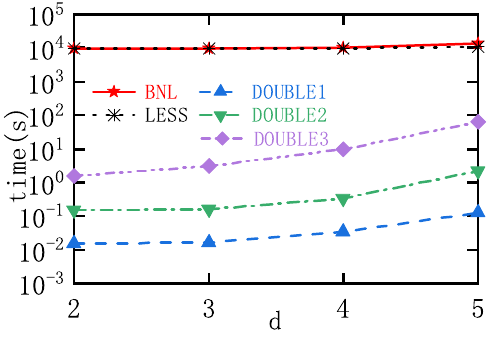}  
	}     
	\subfigure[I/O Overhead] {
		\label{fig3b}
		\includegraphics[width=0.45\linewidth]{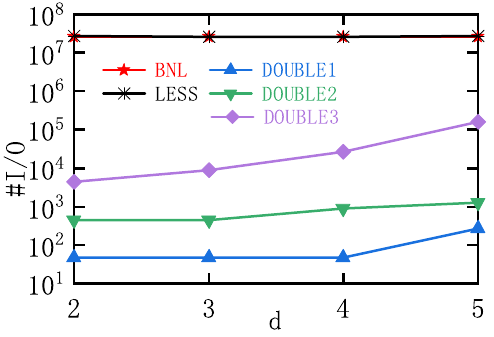}  
	}     
	\subfigure[Used Sample Size] { 
		\label{fig3c}     
		\includegraphics[width=0.45\linewidth]{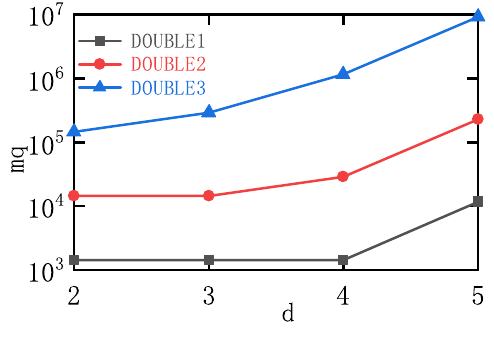}     
	}    
	\subfigure[Result's Representative Ratio] { 
		\label{fig3d}     
		\includegraphics[width=0.45\linewidth]{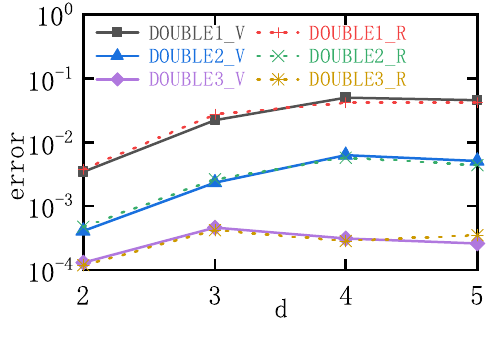}     
	}
	\caption{The Effect of Skyline Criterion Size}
\end{figure}
In experiment 2.2, we validate the effect of skyline criterion size on DOUBLE, with a fixed data volume of 0.2TB. The algorithms are excuted on independently distributed data set. The numbers of skyline criteria considered are 2, 3, 4, 5. As shown in figure \ref{fig3a}, the execution time of DOUBLE grows super linearly with the increase of criterion size. As the dimensionality increases, the required sample grows and exceeds the initial sample, which causes the growth of the execution time. However, in the case of high dimensions, the calculation of the skyline is almost meaningless. Moreover, even with 5 skyline criteria, DOUBLE3 is more than 2 orders of magnitude faster than LESS and BNL. The execution time of LESS and BNL grows slightly with the increase of the criterion size. An increase in $d$ leads to an increase in skyline cardinality, thereby increasing the CPU overheads of LESS and BNL. However, I/O overhead rather than CPU overhead is dominant for LESS and BNL. And the I/O overheads of LESS and BNL is not sensitive to the increase of criterion size, as shown in figure \ref{fig3b}. With 2, 3 and 4 skyline criteria, DOUBLE1 has the same I/O overhead. At this moment, the initial sample is sufficient to cope with the error bound. With more skyline criteria, I/O overhead increases. This is because the required sample size is larger than before and exceeds the initial sample size. Figure \ref{fig3c} shows the actual size of sample required by DOUBLE with respect to the criterion size. As the dimensionality increases, the required sample size grows and surpasses the initial sample size, then the actual sample size becomes larger. Figure \ref{fig3d} shows that the verified error of the returned set is extremely close to the real error. 

\subsubsection{Experiment 2.3: The Effects of Correlation and Anti-correlation}
\begin{figure}[t]
	\centering
	\subfigure[Execution Time] {
		\label{fig4a}
		\includegraphics[width=0.45\linewidth]{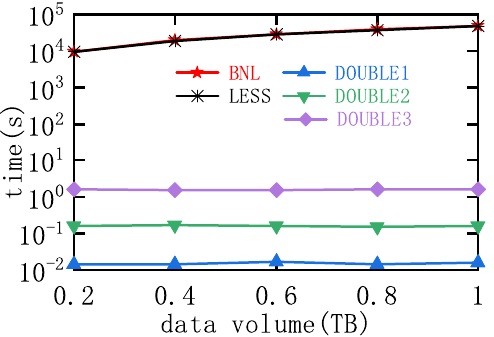}  
	}     
	\subfigure[I/O Overhead] {
		\label{fig4b}
		\includegraphics[width=0.45\linewidth]{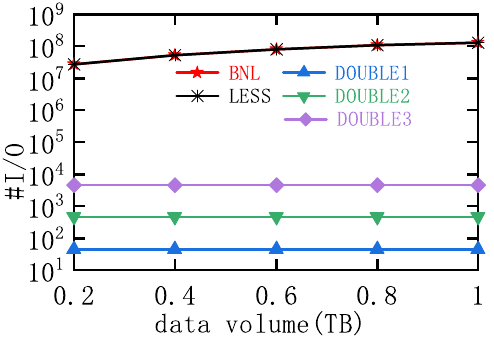}  
	} 
	\caption{The Effects of Correlation}
	\label{fig4}
\end{figure}
\begin{figure}[t]
	\centering
	\subfigure[Execution Time] {
		\label{fig4c}
		\includegraphics[width=0.45\linewidth]{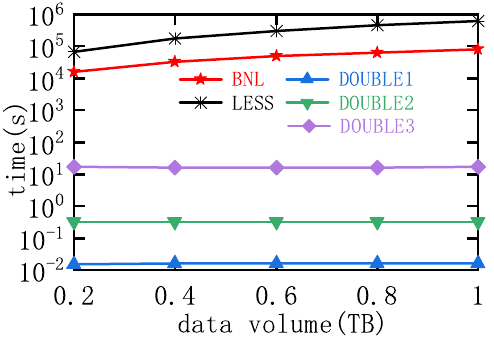}  
	}     
	\subfigure[I/O Overhead] {
		\label{fig4d}
		\includegraphics[width=0.45\linewidth]{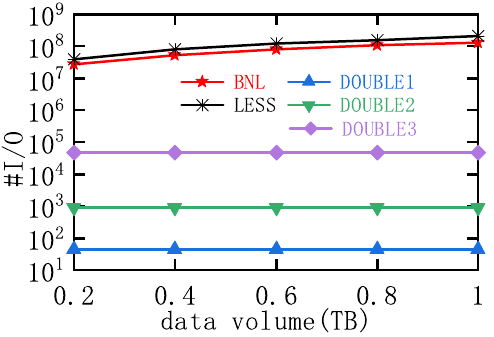}  
	} 
	\caption{The Effects of Anti-correlation}
\end{figure}
In experiment 2.3, we validate the effect of correlation and anti-correlation between attributes. Algorithms are executed on data sets under correlated and anti-correlated distribution with variable data volumes and 2 skyline criteria. As shown in figure \ref{fig4a} and \ref{fig4b}, for BNL and LESS, there is no difference in the execution time and I/O overheads between independent and correlated distributions. The same is true for DOUBLE1, DOUBLE2 and DOUBLE3. At this time, I/O overhead rather than CPU overhead is dominant for LESS and BNL. And the initial sample is sufficient to meet the error bound for DOUBLE. As shown in figure \ref{fig4c}, execution time of the two algorithms is significantly longer under anti-correlated distribution. At this time, CPU overhead is remarkably increased and begin to dominate. And LESS, which has a higher CPU overhead, is obviously not as good as BNL. Under anti-correlated distribution, DOUBLE2 and DOUBLE3 have higher execution time. At this time, the initial sample can no longer meet the demand. However, even under anti-correlated distribution, the execution time of DOUBLE3 is only about ten seconds. As shown in figure \ref{fig4d}, under anti-correlated distribution, except for LESS, the changes in I/O overhead of all algorithms are basically the same as the changes in runtime. The CPU overhead of LESS is dominant at this time, and its change is more obvious than that of I/O overhead. 

\subsubsection{Experiment 2.4: Real Data set}
\begin{figure}[t]
	\centering
	\subfigure[Execution Time] {
		\label{fig5a}
		\includegraphics[width=0.45\linewidth]{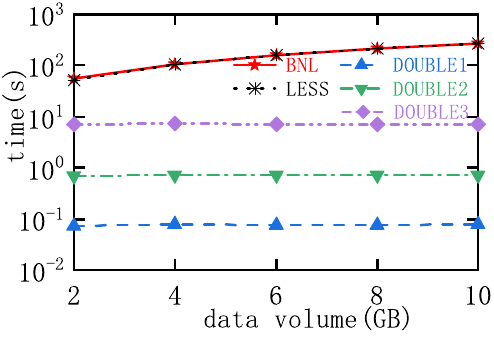}  
	}     
	\subfigure[I/O Overhead] {
		\label{fig5b}
		\includegraphics[width=0.45\linewidth]{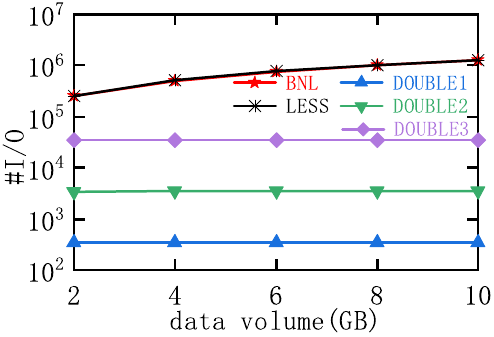}  
	} 
	\caption{Experiment on Real Data set}
\end{figure}
In experiment 2.4, we verify the effectiveness of DOUBLE on the real data set with variable data volumes and 2 skyline criteria. The considered data volumes are 2GB, 4GB, 6GB, 8GB and 10GB. The experimental results on the real data set are almost similar to those on the independent data set. As shown in figure \ref{fig5a}, DOUBLE3 runs nearly 1 orders of magnitude faster than LESS on average. At first glance, compared with independent distribution, the gap becomes smaller on the real data set. This is just because the real data set is much smaller than the synthetic data sets. As the amount of data increases, the gap is incrementally larger. The execution time of LESS and BNL increase linearly with the data volume. For DOUBLE, the execution time is basically not affected by the scale of data. the execution time of DOUBLE3  is a few seconds. As shown in figure \ref{fig5b}, for DOUBLE, I/O overhead has nothing to do with the data volume. But I/O overheads of LESS and BNL grow significantly. For space limitations, the paper do not illustrate the extra experiment statistics of DOUBLE on the real data set. 

\subsection{Summary}
In the experiments, compared to BNL and LESS, DOUBLE runs up to 4 orders of magnitude faster and retrieves up to 4 orders of magnitude fewer disk-pages. Both the execution time and I/O overhead of 3PHASE\_R has nothing to do with the data volume, while these of BNL and LESS are increased linearly, which causes the superiority of DOUBLE increasely obvious with the growth of the data volume. The ascendency of DOUBLE is noticeable. The efficiency of DOUBLE makes it even fully qualified on big data for interactive systems. These all reflect the value of DOUBLE in calculating the skyline on big data. 

When the size of skyline criteria increases, the execution time of DOUBLE increase moderately. Even so, DOUBLE is much faster than all other algorithms. Anti-correlation between attributes is similar to growth of the criterion size for DOUBLE, which increases the required sample size and the skyline cardinality of the sample. Even on anti-correlated distribution data set with 5 skyline criteria, the required sample size still does not exceed ten million, which is relatively small on big data. On distinct data sets, the experimental results verify the theoretical analysis in section \ref{sec_baseline}. It is shown that there is a inseparable relationship between the skyline cardinality and the expected error of the baseline algorithm, and the standard deviation is relately small. 

\section{Conclusion}
\label{sec_con}
In this paper, we proposed two sampling-based approximate algorithms for processing skyline queries on big data. The first algorithm draws a random sample of size $m$ and computes the approximate skyline on the sample. The expected error of the algorithm is almost independent of the input relation size and the standard deviation of the error is relatively small. The running time of the algorithm is $O(m\log^{d-2}m)$ in the worst case and $O(m)$ in the average case. Experiments show that with a moderate size sample, the algorithm has a low enough error. Given $\epsilon$ and $\delta$, the second algorithm returns an ($\epsilon,\delta$)-approximation of the exact skyline. The expected time complexity of the algorithm is  $O(\mathcal{M}_{\frac{\epsilon}{3},\delta}\log^{d-1}\mathcal{M}_{\frac{\epsilon}{3},\delta})$, where $\mathcal{M}_{\frac{\epsilon}{3},\delta}$ is the size of sample required by the first algorithm to return an ($\frac{\epsilon}{3},\delta$)-approximation. $\mathcal{M}_{\frac{\epsilon}{3},\delta}$ is up-bounded by $O(\mathcal{M}_{\epsilon,\delta})$ in most cases, and is almost unaffected by the relation size. Experiments show that the second algorithm is much faster than the existing skyline algorithms. 

\bibliographystyle{splncs04}
\bibliography{sample}

\end{document}